\documentclass[11pt]{article}
\usepackage{geometry}
\usepackage{amsmath, amssymb, amsthm}
\usepackage{graphicx}
\usepackage{bm}
\usepackage[linesnumbered,ruled]{algorithm2e}
\usepackage{fullpage}
\newtheorem{theorem}{Theorem}
\newtheorem{lemma}{Lemma}

\newtheorem{claim}{Claim}
\SetKwRepeat{Do}{do}{while}

\title{Modifying an Instance of the Super-Stable Matching Problem}
\author{Naoyuki Kamiyama%
\thanks{This work was supported by JSPS KAKENHI Grant Number JP20H05795 and 
JST ERATO Grant Number JPMJER2301, Japan.}}
\date{\small Institute of Mathematics for Industry, Kyushu University, Fukuoka, Japan\\
{\ttfamily kamiyama@imi.kyushu-u.ac.jp}}

\begin{document}

\maketitle

\begin{abstract}
Super-stability is one of the stability concepts in the stable matching 
problem with ties. 
It is known that 
there may not exist a super-stable matching, and 
the existence of a super-stable matching can be checked 
in polynomial time. 
In this paper, 
we consider the problem of modifying an instance of 
the super-stable matching problem by deleting some bounded number of 
agents in such a way that  
there exists a super-stable matching in the modified instance. 
First, we prove that if we are allowed to delete agents on only one side, 
then our problem can be solved in polynomial time. 
Interestingly, this result is obtained by carefully observing 
the existing algorithm for checking the existence 
of a super-stable matching. 
In addition, we prove that if we are allowed to delete agents on both sides,
then our problem is NP-complete.
\end{abstract} 

\section{Introduction}

The topic of this paper is the stable matching problem.
This problem was introduced by Gale and Shapley~\cite{GaleS62}, and 
it is one of the most famous mathematical models 
of matching. 
In the basic setting of the stable matching problem, 
it is assumed that 
each agent has a strict preference, i.e., 
the preference of an agent does not contain a tie. 
By contrast, we consider the stable matching problem 
with ties. 
In this paper, we especially focus on super-stability in the stable matching 
problem with ties. 
Super-stability is one of the stability concepts in the stable matching 
problem with ties (see, e.g., \cite{IwamaM08,Manlove13}).
Roughly speaking, super-stability guarantees that 
there does not exist an unmatched pair of agents 
such that the new partners are not worse than their current partners.
It is known that 
there may not exist a super-stable matching, and 
the existence of a super-stable matching can be checked 
in polynomial time~\cite{Irving94, Manlove99}. 
Furthermore, the problem of checking the existence of a 
super-stable matching was considered in  
the many-to-one setting~\cite{IrvingMS00},
the many-to-many setting~\cite{Scott05},
the setting with matroid constraints~\cite{Kamiyama15,Kamiyama19,Kamiyama20},
the setting with master lists~\cite{IrvingMS08,OMalley07}, and 
the student-project allocation probelm~\cite{OlaosebikanM22}. 

The aim of this paper is to consider how to cope with an instance of 
the stable matching problem where 
there does not exist a super-stable matching. 
More precisely, in this paper, 
we consider the problem of modifying an instance of 
the super-stable matching problem by deleting some bounded number of 
agents in such a way that 
there exists a super-stable matching in the modified instance. 
Similar problems for popular matchings 
were considered in \cite{WuLWC13,WuLWC14}. 

Our contribution is summarized as follows. 
First, we
prove that if we are allowed to delete agents on only one side, 
then our problem can be solved in polynomial time (see Section~\ref{section:one_side}). 
Interestingly, this positive result is obtained by carefully observing 
the existing algorithm~\cite{Irving94, Manlove99} for checking the existence 
of a super-stable matching. 
Next, we prove that if we are allowed to delete agents on both sides,
then our problem is NP-complete (see Section~\ref{section:both_sides}).

\section{Preliminaries} 

For each positive integer $z$, 
we define $[z] := \{1,2,\dots,z\}$. 
Define $[0] := \emptyset$. 

In this paper, we are given a finite simple undirected bipartite graph $G = (V, E)$ such that 
the vertex set $V$ is partitioned into $D$ and $H$,
and each edge in $E$ connects a vertex in $D$ and a vertex in $H$. 
We call a vertex in $D$ (resp.\ $H$) a \emph{doctor} (resp.\ \emph{hospital}). 
For each doctor $d \in D$ and each hospital $h \in H$, 
if there exists an edge in $E$ connecting $d$ and $h$, then 
$(d,h)$ denotes this edge.
For each subset $F \subseteq E$ and each doctor $d \in D$
(resp.\ hospital $h \in H$), we define 
$F(d)$ (resp.\ $F(h)$) as the set of edges $(d^{\prime},h^{\prime}) \in F$
such that 
$d^{\prime} = d$ (resp.\ $h^{\prime} = h$). 
Furthermore, for each subset $F \subseteq E$ and each subset $X \subseteq H$, 
we define $F(X) := \bigcup_{h \in X}F(h)$. 
For each subset $X \subseteq V$, 
we define $G\langle X \rangle$ 
as the subgraph of $G$ induced by 
$X$.
For each subset $X \subseteq V$, 
we define $E\langle X \rangle$ as 
the edge set of $G\langle X \rangle$.
For each subset $F \subseteq E$ and 
each vertex $v \in V$ such that $|F(v)| = 1$,  
we do not distinguish between $F(v)$ and the unique edge
in $F(v)$. 

For each vertex $v \in V$, we are given a transitive 
binary relation $\succsim_v$ on $E(v) \cup \{\emptyset\}$ 
satisfying the following conditions. 
\begin{itemize}
\item
For every pair of elements $e, f \in E(v) \cup \{\emptyset\}$,
at least one of $e \succsim_v f$ and $f \succsim_v e$ holds.  
\item
For every edge $e \in E(v)$, we have $e \succsim_v \emptyset$ and 
$\emptyset \not\succsim_v e$. 
\end{itemize}
For each vertex $v \in V$ and 
each pair of edges $e,f \in E(v)$, 
we write $e \succ_v f$ if $e \succsim_v f$ and 
$f \not\succsim_v e$. 

Let $X$ be a subset of $V$.
A subset $\mu \subseteq E\langle X \rangle$ 
is called a \emph{matching in $G\langle X \rangle$} 
if $|\mu(v)| \le 1$ holds for  
every vertex $v \in X$. 
For each matching $\mu$ in $G\langle X \rangle$ 
and each edge $e = (d,h) \in E\langle X \rangle \setminus \mu$, 
we say that $e$ \emph{blocks} $\mu$ if 
$e \succsim_d \mu(d)$ and 
$e \succsim_h \mu(h)$. 
A matching $\mu$ in $G\langle X \rangle$ 
is said to be \emph{super-stable} if there
does not exist an edge in $E\langle X \rangle \setminus \mu$ blocking $\mu$.

For each subset $X \subseteq V$, $X$ is said to be \emph{feasible}
if there exists 
a super-stable matching in $G\langle V \setminus X \rangle$. 
Define $\mathcal{F}$ as the family of 
feasible subsets of $V$. 
Furthermore, we define $\mathcal{F}_H$ as the set of elements 
$X \in \mathcal{F}$ such that $X \subseteq H$. 
In this paper, we consider the following problems. 
\begin{description}
\item[Problem~1.]
We are given a non-negative integer $q$.
The goal of this problem is to determine whether 
there exists an element $X \in \mathcal{F}_H$
such that $|X| \le q$.
\item[Problem~2.] 
We are given non-negative integers $q_1,q_2$.
The goal of this problem is to determine whether 
there exists an element $X \in \mathcal{F}$
such that $|X \cap D| \le q_1$ and $|X \cap H| \le q_2$.
\end{description}

In this paper, we first prove that 
Problem~1 can be solved in polynomial time (see Section~\ref{section:one_side}). 
Furthermore, we prove that Problem~2 is NP-complete (see Section~\ref{section:both_sides}). 

\section{Deleting Agents on One Side} 
\label{section:one_side} 

The goal of this section is to propose a polynomial-time 
algorithm for Problem~1. 

For each doctor $d \in D$ and
each subset $F \subseteq E$, 
we define ${\rm Ch}_d(F)$ 
as the set of edges $e \in F(d)$ such that 
$e \succsim_d f$ for every edge $f \in F(d)$. 
For each hospital $h \in H$ and 
each subset $F \subseteq E$, 
we define ${\rm Ch}_h(F)$ 
as the set of edges $e \in F(h)$ such that 
$e \succ_h f$ for every edge $f \in F(h)$.
Notice 
that ${\rm Ch}_h(F)$ may be empty. 
For each subset $F \subseteq E$, we define 
${\rm Ch}_D(F) := \bigcup_{d \in D}{\rm Ch}_d(F)$ and 
${\rm Ch}_H(F) := \bigcup_{h \in H}{\rm Ch}_h(F)$.
For each subset $F \subseteq E$, we define $D[F]$ as the set of doctors 
$d \in D$ such that $F(d) \neq \emptyset$.  

For each subset $X \subseteq H$, we define 
the subset ${\sf B}_X \subseteq E$ as the output of
Algorithm~\ref{alg:pre-process}.
Notice that although Algorithm~\ref{alg:pre-process} looks different from 
the algorithm in \cite{Irving94, Manlove99}, 
Algorithm~\ref{alg:pre-process} is basically  
the same as the algorithm in \cite{Irving94, Manlove99}. 
The main 
contribution of this section is
to reveal that the algorithm in \cite{Irving94, Manlove99}
gives more information about  
super-stable matchings. 

\begin{algorithm}[h]
Define $R_{X,0} := E(X)$.\\
Set $t := 0$.\\
\Do{$R_{X,t-1} \neq R_{X,t}$}
{
  Set $t := t + 1$.\\
  Define $P_{X,t} := {\rm Ch}_D(E \setminus R_{X,t-1})$.\\
  Define $Q_{X,t} := {\rm Ch}_H(P_{X,t} \cup R_{X,t-1}) \cap P_{X,t}$.\\
  Define $R_{X,t} := R_{X,t-1} \cup (P_{X,t} \setminus Q_{X,t})$.   
}
Define $k_X := t$.\\
Output $R_{X,k_X}$ as ${\sf B}_X$, and halt.
\caption{Algorithm for defining ${\sf B}_X$}
\label{alg:pre-process}
\end{algorithm}

\begin{lemma} \label{lemma:pre-process_iteration}
The number of iterations of Steps~3 to 8 of Algorithm~\ref{alg:pre-process}
is at most $|E|$.
\end{lemma}
\begin{proof}
This lemma immediately follows from the fact that 
$R_{X,t-1} \subseteq R_{X,t}$.
\end{proof} 

Lemma~\ref{lemma:pre-process_iteration} implies that 
Algorithm~\ref{alg:pre-process} is a polynomial-time 
algorithm. 

\begin{lemma} \label{lemma:difference}
Let $X$ be a subset of $H$. 
Then for every doctor $d \in D$,  
every integer $t \in \{0\} \cup [k_X]$, and every pair 
of edges $e \in R_{X,t}(d) \setminus E(X)$ and $f \in E(d) \setminus R_{X,t}$,  
we have $e \succsim_d f$. 
\end{lemma}
\begin{proof}
Let $d$ be a doctor in $D$. 
Since $R_{X,0} = E(X)$, this lemma holds when $t = 0$. 
Let $n$ be an integer in $[k_X]$. 
Assume that this lemma holds when $t = n-1$.
Then since 
$R_{X,n} \setminus R_{X,n-1}$
is a subset of ${\rm Ch}_D(E \setminus R_{X,n-1})$, 
this lemma holds when $t = n$. 
This completes the proof. 
\end{proof} 

\begin{lemma} \label{lemma:useful}
For every subset $X \subseteq H$, the following statements 
hold. 
\begin{description}
\item[\bf (R1)]
For every 
super-stable matching $\mu$
in $G\langle V \setminus X \rangle$, 
we have $\mu \subseteq E \setminus {\sf B}_X$.
\item[\bf (R2)]
${\sf B}_{\emptyset} \subseteq {\sf B}_X$.  
\end{description}
\end{lemma} 
We give the proof of Lemma~\ref{lemma:useful} in Section~\ref{section:proof}. 
Although the proof of (R1) was given in \cite{Irving94, Manlove99}, 
we give its proof for completeness.
To the best of our knowledge, (R2) is a new observation. 

\begin{lemma} \label{lemma:pre-process_matching}
Let $X$ be a subset of $H$.
Then 
there exists a matching $\mu_X$ in $G\langle V \setminus X \rangle$ 
such that
$\mu_X \subseteq {\rm Ch}_D(E \setminus {\sf B}_X)$ and 
$\mu_X(d) \neq \emptyset$ holds 
for every doctor $d \in D[E \setminus {\sf B}_X]$.
\end{lemma}
\begin{proof}
In this proof, we define 
$P := P_{X,k_X}$,  
$Q := Q_{X,k_X}$,  
$R := R_{X,k_X}$, and  
$R_{-1} := R_{X,k_X-1}$
for notational simplicity. 

Since ${\sf B}_X = R = R_{-1}$, 
we have $P = {\rm Ch}_D(E \setminus {\sf B}_X)$. 
Thus, if $|P(h)| \le 1$ for every hospital $h \in H$, 
then 
a matching $\mu_X$ satisfying the conditions in this lemma 
can be obtained by setting 
$\mu_X(d)$ to be an arbitrary edge in $P(d)$ 
for each doctor $d \in D[E \setminus {\sf B}_X]$. 

Assume that there exists a hospital $h \in H$ such that 
$|P(h)| \ge 2$.
Then since the definition of ${\rm Ch}_h(\cdot)$ implies that 
${\rm Ch}_H(P \cup R_{-1})$ contains at most one edge in 
$P(h)$, 
$P \setminus Q$ is not empty. 
However, this contradicts the fact that 
$R_{-1} = R$.
This completes the proof.
\end{proof} 

Let $X$ be a subset of $H$. 
For each matching $\mu_X$ in $G\langle V \setminus X \rangle$
satisfying the conditions in Lemma~\ref{lemma:pre-process_matching}, 
a hospital $h \in H$ is said to be 
\emph{critical} with respect to $({\sf B}_X,\mu_X)$ if 
(i) $\mu_X(h) = \emptyset$, and (ii) 
at least one of ${\sf B}_X(h) \neq \emptyset$ and 
${\rm Ch}_D(E \setminus {\sf B}_X) \cap E(h) \neq \emptyset$ holds. 

In what follows, let $\mu_{\emptyset}$ be an arbitrary matching in $G$ 
satisfying the conditions in Lemma~\ref{lemma:pre-process_matching}. 
Define $S_{\emptyset}$ as the set of critical hospitals 
with respect to $({\sf B}_{\emptyset},\mu_{\emptyset})$.  

\begin{lemma} \label{lemma:upper_bound} 
$S_{\emptyset} \in \mathcal{F}_H$. 
\end{lemma}
\begin{proof}
In this proof, we define 
$P := P_{\emptyset,k_{\emptyset}}$,  
$Q := Q_{\emptyset,k_{\emptyset}}$,  
$R := R_{{\emptyset},k_{\emptyset}}$, 
$R_{-1} := R_{\emptyset,k_{\emptyset}-1}$,
$\mu := \mu_{\emptyset}$, and  
$S := S_{\emptyset}$
for notational simplicity. 

Since the definition of a critical hospital implies that 
$\mu(h) = \emptyset$ for every hospital $h \in S$,
$\mu$ is a matching in $G \langle V \setminus S \rangle$.
Thus, we prove that $\mu$ is a super-stable matching in 
$G \langle V \setminus S \rangle$. 

Let $e = (d,h)$ be an edge in 
$E\langle V \setminus S \rangle \setminus \mu$.
Then we prove that $e$ does not block $\mu$.
Clearly, if $\mu(d) \succ_d e$, then 
the proof is done. 
Thus, we assume that 
$e \succsim_d \mu(d)$. 

We first consider the case where 
$e \in {\sf B}_{\emptyset}$. 
If $\mu(h) = \emptyset$, then 
since $e \in {\sf B}_{\emptyset}(h)$, 
we have $h \in S$. 
However, this contradicts the fact that $e \notin E(S)$.
Thus, we can assume that 
$\mu(h) \neq \emptyset$.
Assume that $e \succsim_h \mu(h)$. 
Then we derive a contradiction. 
Since $\mu$ satisfies the conditions in 
Lemma~\ref{lemma:pre-process_matching}, 
$\mu \subseteq {\rm Ch}_D(E \setminus {\sf B}_{\emptyset})$. 
Thus, since 
${\sf B}_{\emptyset} = R = R_{-1}$, 
we have 
$\mu(h) \in P = {\rm Ch}_D(E \setminus R_{-1})$. 
In addition, since $e \in {\sf B}_{\emptyset} = R_{-1}$ and $e \succsim_h \mu(h)$, 
we have 
$\mu(h) \notin Q$. 
Thus, we have
$\mu(h) \in P \setminus Q$. 
However, this contradicts the fact that 
$R_{-1} = R$. 

Next we consider the case where 
$e \notin {\sf B}_{\emptyset}$. 
Since $\mu(d) \in {\rm Ch}_D(E \setminus {\sf B}_{\emptyset})$
and $e \succsim_d \mu(d)$, 
we have 
$e \in {\rm Ch}_D(E \setminus {\sf B}_{\emptyset})$. 
Thus, if $\mu(h) = \emptyset$, then 
$h \in S$. 
This contradicts the fact that $e \in E(S)$.
Thus, we can assume that 
$\mu(h) \neq \emptyset$. 
Assume that $e \succsim_h \mu(h)$. 
Then we derive a contradiction. 
Since $\mu$ satisfies the conditions in 
Lemma~\ref{lemma:pre-process_matching}, 
we have $\mu(h) \in P$. 
In addition, since $e \in {\rm Ch}_D(E \setminus {\sf B}_{\emptyset}) = P$
and $e \succsim_h \mu(h)$, 
we have $\mu(h) \notin Q$. 
Thus, 
$\mu(h) \in P \setminus Q$. 
However, this contradicts the fact that 
$R_{-1} = R$. 
This completes the proof. 
\end{proof} 

\begin{lemma} \label{lemma:lower_bound} 
For every element $X \in \mathcal{F}_H$, 
we have $|S_{\emptyset}| \le |X|$. 
\end{lemma}
\begin{proof}
Assume that there exists an element $X \in \mathcal{F}_H$ 
such that $|S_{\emptyset}| > |X|$. 
Then since $X \in \mathcal{F}_H$, 
there exists a super-stable matching $\sigma$ in $G \langle V \setminus X \rangle$. 
Notice that (R1) of Lemma~\ref{lemma:useful} implies that 
$\sigma \subseteq E \setminus {\sf B}_X$. 

\begin{claim} \label{claim_1:lemma:lower_bound}
$|\mu_{\emptyset}| \ge |\sigma|$. 
\end{claim}
\begin{proof}
Recall that 
$\mu_{\emptyset}(d) \neq \emptyset$ holds for every 
doctor $d \in D[E \setminus {\sf B}_{\emptyset}]$. 
Thus, in order to prove this claim, 
it is sufficient to prove that, for every doctor $d \in D$ such that 
$E(d) \subseteq {\sf B}_{\emptyset}$, $\sigma(d) = \emptyset$.
Let $d$ be a doctor in $D$ such that 
$E(d) \subseteq {\sf B}_{\emptyset}$.
Then (R2) of Lemma~\ref{lemma:useful} implies that 
$E(d) \subseteq {\sf B}_X$. 
Thus, since $\sigma \subseteq E \setminus {\sf B}_X$, 
we have $\sigma(d) = \emptyset$. 
This completes the proof. 
\end{proof}

\begin{claim} \label{claim_2:lemma:lower_bound}
For every hospital $h \in S_{\emptyset} \setminus X$, 
we have $\sigma(h) \neq \emptyset$. 
\end{claim}
\begin{proof}
Assume that there exists a hospital 
$h \in S_{\emptyset} \setminus X$
such that $\sigma(h) = \emptyset$. 

We first consider the case where 
there exists an edge $e = (d,h) \in {\sf B}_{\emptyset}$. 
In this case, it follows from (R2) of Lemma~\ref{lemma:useful} that 
$e \in {\sf B}_X$. 
Since $h \notin X$, we have 
$e \notin E(X)$. 
Thus, since $\sigma \subseteq E \setminus {\sf B}_X$, 
Lemma~\ref{lemma:difference} implies that 
$e \succsim_d \sigma(d)$. 
This implies that since $\sigma(h) = \emptyset$,  
$e$ blocks $\sigma$. 
This contradicts the fact that $\sigma$ is a super-stable matching in 
$G\langle V \setminus X \rangle$. 

Next, we consider the case where 
there exists an edge $e = (d,h) \in 
{\rm Ch}_D(E \setminus {\sf B}_{\emptyset})$. 
If $e \in {\sf B}_X$, then 
we can prove this case in the same way as the first case. 
Thus, we can assume that 
$e \notin {\sf B}_X$. 
Then since 
$e \in {\rm Ch}_D(E \setminus {\sf B}_{\emptyset})$ and 
${\sf B}_{\emptyset} \subseteq {\sf B}_X$, 
we have 
$e \in {\rm Ch}_D(E \setminus {\sf B}_X)$.
Thus,
since $\sigma \subseteq E \setminus {\sf B}_X$, 
we have $e \succsim_d \sigma(d)$. 
Since $\sigma(h) = \emptyset$,  
$e$ blocks $\sigma$. 
This contradicts the fact that $\sigma$ is a super-stable matching in 
$G\langle V \setminus X \rangle$. 
This completes the proof. 
\end{proof}

Define 
$H_{\emptyset} := \{h \in H \mid \mu_{\emptyset}(h) \neq \emptyset\}$
and 
$H_{X} := \{h \in H \setminus X \mid \sigma(h) \neq \emptyset\}$.
Notice that 
since $H_X \subseteq H \setminus X$, 
we have 
$H_X \cap H_{\emptyset} \subseteq H_{\emptyset} \setminus X$. 
Here we prove that 
$H_X \cap H_{\emptyset} \subsetneq H_{\emptyset} \setminus X$. 
Assume that 
$H_X \cap H_{\emptyset} = H_{\emptyset} \setminus X$. 
Claim~\ref{claim_1:lemma:lower_bound} implies that 
$|H_{\emptyset}| \ge |H_X|$. 
Furthermore, Claim~\ref{claim_2:lemma:lower_bound} 
implies that 
$H_X \cap S_{\emptyset} = S_{\emptyset} \setminus X$.
Thus, since $H_{\emptyset} \cap S_{\emptyset} = \emptyset$ and 
$|S_{\emptyset}| > |X|$, we have
\begin{equation*}
\begin{split}
& |H_{\emptyset}| + |S_{\emptyset}| > |H_X| + |X| 
\ge |H_X \cap H_{\emptyset}| + |H_X \cap S_{\emptyset}| + |X|
= |H_{\emptyset} \setminus X| + |S_{\emptyset} \setminus X| + |X|\\
& = |H_{\emptyset}| - |H_{\emptyset} \cap X| 
+ |S_{\emptyset}| - |S_{\emptyset} \cap X| + |X|
\ge |H_{\emptyset}| + |S_{\emptyset}| - |X| + |X|
= |H_{\emptyset}| + |S_{\emptyset}|.
\end{split} 
\end{equation*}
However, this is a contradiction. 
Thus, we have 
$H_X \cap H_{\emptyset} \subsetneq H_{\emptyset} \setminus X$.

Since 
$H_X \cap H_{\emptyset} \subsetneq H_{\emptyset} \setminus X$, 
there exists a hospital $h \in H_{\emptyset} \setminus X$ and 
$h \notin H_X$, i.e.,  
there exists a hospital $h \in H \setminus X$ such that 
$\mu_{\emptyset}(h) \neq \emptyset$ and 
$\sigma(h) = \emptyset$. 
Assume that $\mu_{\emptyset}(h) = (d,h)$.

If $\mu_{\emptyset}(h) \in {\sf B}_X$, then 
since $\mu_{\emptyset}(h) \notin E(X)$ and 
$\sigma \subseteq E \setminus {\sf B}_X$, 
Lemma~\ref{lemma:difference} implies that 
$\mu_{\emptyset}(h) \succsim_d \sigma(d)$. 
If 
$\mu_{\emptyset}(h) \notin {\sf B}_X$, then 
since $\mu_{\emptyset}(h) \in {\rm Ch}_D(E \setminus {\sf B}_{\emptyset})$ and 
${\sf B}_{\emptyset} \subseteq {\sf B}_X$, 
we have 
$\mu_{\emptyset}(h) \in {\rm Ch}_D(E \setminus {\sf B}_X)$. 
This implies that 
since $\sigma \subseteq E \setminus {\sf B}_X$, 
$\mu_{\emptyset}(h) \succsim_d \sigma(d)$. 
That is, in both cases, 
$\mu_{\emptyset}(h) \succsim_d \sigma(d)$. 
Thus, since 
$\sigma(h) = \emptyset$, 
$\mu_{\emptyset}(h)$ blocks $\sigma$.
However, this contradicts the fact that $\sigma$ is a super-stable matching in 
$G\langle V \setminus X \rangle$. 
This completes the proof. 
\end{proof} 

\begin{theorem} \label{theorem:problem_1}
There exists an element $X \in \mathcal{F}_H$ such that 
$|X| \le q$
if and only if $|S_{\emptyset}| \le q$. 
\end{theorem}
\begin{proof}
This theorem immediately follows from Lemmas~\ref{lemma:upper_bound} and 
\ref{lemma:lower_bound}. 
\end{proof} 

By using Theorem~\ref{theorem:problem_1}, 
we can solve Problem~1 in polynomial time. 
First, we compute ${\sf B}_{\emptyset}$ and $\mu_{\emptyset}$ in polynomial time.
Then we determine whether $|S_{\emptyset}| \le q$. 

Here we give a remark on the algorithm in 
\cite{Irving94, Manlove99}.
First, the algorithm in 
\cite{Irving94, Manlove99} determines whether $S_{\emptyset} = \emptyset$. 
If $S_{\emptyset} = \emptyset$, then the algorithm in 
\cite{Irving94, Manlove99} concludes that 
there exists a super-stable matching. 
Otherwise, it concludes that there does not exist a super-stable matching. 

\subsection{Proof of Lemma~\ref{lemma:useful}}
\label{section:proof} 

Let $X$ be a subset of $H$. 
Lemma~\ref{lemma:useful} follows from 
the following lemmas.

\begin{lemma} \label{lemma:R1}
For every 
super-stable matching $\mu$
in $G\langle V \setminus X \rangle$, 
we have $\mu \subseteq E \setminus {\sf B}_X$.
\end{lemma} 
\begin{proof}
For notational simplicity, in this proof, we define 
$R_t := R_{X,t}$ for each integer $t \in [k_X]$. 

An edge $e \in {\sf B}_X$ is called a \emph{bad edge} if 
there exists a super-stable matching in $G\langle V \setminus X \rangle$ that 
contains $e$. 
If we can prove that there does not exist a bad edge 
in ${\sf B}_X$, then the proof is done. 
Assume that there exists a bad edge in ${\sf B}_X$. 
Define $\Delta$ as the set of integers $t \in [k_X]$ such that 
$R_{t} \setminus R_{t-1}$ contains 
a bad edge in ${\sf B}_X$. 
Let $z$ be the minimum integer in $\Delta$. 
Then $R_{z-1}$ does not contain a bad edge in ${\sf B}_X$.

Let $e = (d,h)$ be a bad edge in $R_{z} \setminus R_{z-1}$.
Then there exists a 
super-stable matching $\sigma$ in $G \langle V \setminus X \rangle$ such that 
$e \in \sigma$. 
Since $e \in R_{z} \setminus R_{z-1}$, 
we have $e \in {\rm Ch}_D(E \setminus R_{z-1})$
and 
there exists an edge $f = (p,h)\in {\rm Ch}_D(E \setminus R_{z-1}) \cup R_{z-1}$
such that $f \neq e$ and $f \succsim_h e$. 
Since $E(X) \subseteq R_{z-1}$, 
we have $e \notin E(X)$. 
Thus, 
$h \notin X$ and $f \notin E(X)$. 

First, we consider the case where 
$f \in {\rm Ch}_D(E \setminus R_{z-1})$.
If $\sigma(p) \succ_p f$, then 
$\sigma(p) \in R_{z-1}$.
In this case, 
$\sigma(p)$ is a bad edge in ${\sf B}_X$.
However,  
this contradicts the minimality of $z$. 
Thus, we can assume that $f \succsim_p \sigma(p)$. 
In this case, 
$f$ blocks $\sigma$. 
However, this contradicts the fact that 
$\sigma$ is a super-stable matching in $G\langle V \setminus X \rangle$. 

Next, we consider the case where 
$f \in R_{z-1}$.
If $\sigma(p) \in R_{z-1}$, then 
$\sigma(p)$ is a bad edge in ${\sf B}_X$.
This contradicts the minimality of $z$.
Thus, we can assume that 
$\sigma(p) \in E(p) \setminus R_{z-1}$.
Then  
since $f \notin E(X)$, 
Lemma~\ref{lemma:difference} implies that 
$f \succsim_p \sigma(p)$. 
Thus,
$f$ blocks $\sigma$. 
However, this contradicts the fact that 
$\sigma$ is a super-stable matching in $G\langle V \setminus X \rangle$. 
This completes the proof. 
\end{proof} 

\begin{lemma} \label{lemma:R2}
${\sf B}_{\emptyset} \subseteq {\sf B}_X$.
\end{lemma} 
\begin{proof}
Let $k := \max\{k_{\emptyset},k_X\}$. 
For each symbol $Z \in \{\emptyset,X\}$, 
if $k_Z < k$, then 
we define 
$P_{Z,t} := P_{Z,k_Z}$,
$Q_{Z,t} := Q_{Z,k_Z}$, and 
$R_{Z,t} := R_{Z,k_Z}$ 
for each integer $t \in [k] \setminus [k_Z]$. 

\begin{claim}
Let $Z$ be a symbol in $\{\emptyset,X\}$.
Then for every integer $t \in [k]$, 
\begin{equation} \label{eq:same}
\begin{split}
P_{Z,t} & = {\rm Ch}_D(E \setminus R_{Z,t-1}), \\ 
Q_{Z,t} & = {\rm Ch}_H(P_{Z,t} \cup R_{Z,t-1}) \cap P_{Z,t}, \\  
R_{Z,t} & = R_{Z,t-1} \cup (P_{Z,t} \setminus Q_{Z,t}).   
\end{split}
\end{equation}
\end{claim}
\begin{proof}
Let $\ell$ be an integer in $[k]$. 
If $\ell \in [k_Z]$, then 
the definition of 
Algorithm~\ref{alg:pre-process}
implies that 
\eqref{eq:same} holds
when $t = \ell$. 
Thus, we consider the case where 
$\ell \notin [k_Z]$. 
Notice that if $\ell = k_Z+1$, then 
the definition of 
Algorithm~\ref{alg:pre-process} 
implies that 
$R_{Z,\ell-1} = R_{Z,\ell-2}$. 
Furthermore, if $\ell > k_Z + 1$, then 
$R_{Z,\ell-1} = R_{Z,\ell-2} = R_{Z,k_Z}$. 
Thus, in both cases, we have 
$R_{Z,\ell-1} = R_{Z,\ell-2}$. 

Assume that 
\eqref{eq:same} holds when $t = \ell -1$.
Since $P_{Z,\ell} = P_{Z,\ell-1} = P_{Z,k_Z}$, 
\begin{equation*}
P_{Z,\ell} = P_{Z,\ell-1}
= 
{\rm Ch}_D(E \setminus R_{Z,\ell-2})
= 
{\rm Ch}_D(E \setminus R_{Z,\ell-1}). 
\end{equation*}
Since $Q_{Z,\ell} = Q_{Z,\ell-1} = Q_{Z,k_Z}$,
\begin{equation*}
Q_{Z,\ell} = Q_{Z,\ell-1}
= {\rm Ch}_H(P_{Z,\ell-1} \cup R_{Z,\ell-2}) \cap P_{Z,\ell-1}
= {\rm Ch}_H(P_{Z,\ell} \cup R_{Z,\ell-1}) \cap P_{Z,\ell}. 
\end{equation*}
Since $R_{Z,\ell} = R_{Z,\ell-1} = R_{Z,k_Z}$,
\begin{equation*}
R_{Z,\ell} = R_{Z,\ell-1}
= R_{Z,\ell-2} \cup (P_{Z,\ell-1} \setminus Q_{Z,\ell-1})
= R_{Z,\ell-1} \cup (P_{Z,\ell} \setminus Q_{Z,\ell}). 
\end{equation*}
This completes the proof. 
\end{proof} 

In what follows, we prove that 
$R_{\emptyset,t} \subseteq R_{X,t}$ 
for every integer $t \in [k]$.
Notice that $R_{\emptyset,0} = \emptyset$. 

Let $t$ be an integer in $\{0\} \cup [k-1]$, and 
we assume that 
$R_{\emptyset,t} \subseteq R_{X,t}$.
Then we prove that 
$R_{\emptyset,t+1} \subseteq R_{X,t+1}$. 
Let $e = (d,h)$ be an edge in 
$R_{\emptyset,t+1} \setminus R_{\emptyset,t}$.
Then 
$e \in {\rm Ch}_D(E \setminus R_{\emptyset,t})$, and 
there exists an edge 
$f \in {\rm Ch}_D(E\setminus R_{\emptyset,t}) \cup R_{\emptyset,t}$
such that $f \neq e$ and $f \succsim_h e$. 
If $e \in R_{X,t}$, then the proof is done. 
Thus, we can assume that 
$e \notin R_{X,t}$. 
Since $R_{\emptyset,t} \subseteq R_{X,t}$, 
$e \in {\rm Ch}_D(E \setminus R_{X,t})$. 
Thus, we prove that 
there exists an edge 
$g \in {\rm Ch}_D(E\setminus R_{X,t}) \cup R_{X,t}$
such that $g \neq e$ and $g \succsim_h e$. 

If $f \in R_{\emptyset,t}$, then 
since $R_{\emptyset,t} \subseteq R_{X,t}$, 
we have $f \in R_{X,t}$.
Thus, the proof is done. 
Assume that $f \in {\rm Ch}_D(E\setminus R_{\emptyset,t})$. 
If $f \in R_{X,t}$, then the proof is done. 
If $f \notin R_{X,t}$, then
since $R_{\emptyset,t} \subseteq R_{X,t}$, we have  
$f \in {\rm Ch}_D(E\setminus R_{X,t})$.
This completes the proof.
\end{proof} 

\section{Deleting Agents on Both Sides} 
\label{section:both_sides} 

The goal of this section is to prove that 
Problem~2 is NP-complete. 

\begin{theorem}
Problem~2 is NP-complete. 
\end{theorem}
\begin{proof}
For every subset $X \subseteq V$, 
we can determine whether there exists a super-stable matching in 
$G\langle V \setminus X \rangle$ in polynomial time 
by using the algorithm in  
\cite{Irving94, Manlove99}.
Thus, Problem~2 is in NP.
In what follows, 
we prove the NP-hardness of Problem~2. 

We prove the ${\rm NP}$-hardness of Problem~2 by reduction 
from the decision version of {\sc Minimum Coverage}. 
In this problem, we are given a finite set $S = \{s_1,s_2,\dots,s_n\}$, 
subsets $T_1, T_2, \dots, T_m$ of $S$, 
and non-negative integers $x, y$ such that $x \le m$ 
and $y \le n$.
Then the goal of this problem is to determine whether there exists a subset 
$I \subseteq [m]$
such that $|I| = x$ and $|\bigcup_{i \in I} T_i| \le y$.
It is known that this problem is ${\rm NP}$-complete~\cite{Vinterbo02}.

Assume that we are given an instance of 
{\sc Minimum Coverage}. 
We construct an instance of Problem~2 as follows. 
Define 
\begin{equation*}
D := \{T_1,T_2,\dots,T_m\}, \ \  
H := \{s_1,s_2,\dots,s_n,t_1,t_2,\dots,t_m\}, 
\end{equation*}
where
$t_1,t_2,\dots,t_m$ are new elements. 
Then for each doctor $T_i \in D$ and
each hospital $s_j \in H$, $E$ contains the edge $(T_i,s_j)$ if and only if 
$s_j \in T_i$. 
For each doctor $T_i \in D$, 
$E$ contains the edge $(T_i,t_i)$.  
Define $q_1 := m - x$ and 
$q_2 := y$. 
Furthermore, for each vertex $v \in V$, we define $\succsim_v$ in such a way that 
$e \succsim_v f$ and 
$f \succsim_v e$
for every pair of edges $e,f \in E(v)$. 
Clearly, this reduction can be done in polynomial time. 

Assume that there exists a feasible solution $I$ to the given instance of 
{\sc Minimum Coverage}. 
Define 
\begin{equation*}
X := \{T_i \mid i \notin I\} 
\cup \{s_j \in S \mid s_j \in \textstyle{\bigcup_{i \in I} T_i}\}. 
\end{equation*}
Then we have 
\begin{equation*}
|D \cap X| = m - |I| = m - x = q_1, \ \ \
|H \cap X| = |\textstyle{\bigcup_{i \in I} T_i}| \le y = q_2.
\end{equation*}
Thus, 
what remains is to prove that there exists a super-stable matching in 
$G\langle V \setminus X \rangle$. 
Define the matching $\mu$ in $G\langle V \setminus X \rangle$ 
by $\mu := \{(T_i,t_i) \mid i \in I\}$. 
Then since $E\langle V \setminus X \rangle \cap E(T_i) = \{(T_i,t_i)\}$ 
for every integer $i \in I$ and $D \setminus X = \{T_i \mid i \in I\}$, 
$\mu$ is super-stable. 

Assume that there exists an element $X \in \mathcal{F}$ such that 
$|D \cap X| \le q_1$ and 
$|H \cap X| \le q_2$. 
In addition, 
we assume that $X$ minimizes 
$|\{t_1,t_2,\dots,t_m\} \cap X|$ among all the elements 
$Y \in \mathcal{F}$ such that 
$|D \cap Y| \le q_1$ and 
$|H \cap Y| \le q_2$. 
Since $X \in \mathcal{F}$, 
there exists a super-stable matching $\mu$ in $G\langle V \setminus X \rangle$. 
Define 
\begin{equation*}
I := \{i \in [m] \mid T_i \notin X\}. 
\end{equation*}
Then we have 
$|I| = m - |D \cap X| \ge m - q_1 = x$.

\begin{claim} \label{claim_1:hardness} 
$\bigcup_{i \in I} T_i \subseteq X$. 
\end{claim}
\begin{proof}
First, we prove that 
$|T_i \setminus X| \le 1$ holds for every integer $i \in I$.
Assume that there exists an integer $i \in I$ such that 
$|T_i \setminus X| \ge 2$.
In this case, even if $\mu(T_i) \neq \emptyset$, 
there exists an element $s \in T_i \setminus X$ such that 
$(T_i, s) \notin \mu$. 
However, this implies that $\mu$ is not super-stable. 

Let $I^{\circ}$ be the set of integers $i \in I$ such that 
$|T_i \setminus X| = 1$. 
For each integer $i \in I^{\circ}$, we define 
$h_i$ as the element in $T_i \setminus X$. 
Since $\mu$ is super-stable, 
$t_i \in X$ and $(T_i,h_i) \in \mu$ for every 
integer $i \in I^{\circ}$. 
Define $\mu^{\circ}$ as the matching in $G$ obtained from $\mu$ 
by replacing $(T_i,h_i)$ with $(T_i,t_i)$ for all the integers $i \in I^{\circ}$.
Then $\mu^{\circ}$ is super-stable, and 
$(X \setminus \{t_i \mid i \in I^{\circ}\}) \cup \{h_i \mid i \in I^{\circ}\}$
is feasible. 
This contradicts the assumption on $X$. 
This completes the proof.
\end{proof}

Claim~\ref{claim_1:hardness} implies that 
$|\bigcup_{i \in I} T_i| \le |H \cap X| \le q_2 = y$. 
Let $J$ be a subset of $I$ such that 
$|J| = x$. 
Then we have $|\bigcup_{i \in J} T_i| \le y$. 
Thus, $J$ is a feasible solution to the given instance 
of {\sc Minimum Coverage}.
This completes the proof. 
\end{proof} 

\bibliographystyle{plain}
\bibliography{super_mod_bib}

\end{document}